\documentclass{article}
\usepackage{amsmath,amssymb}

\def\naturals{\mathbb{N}}
\def\upto{{\raise 1pt \hbox{$\scriptstyle \,\nearrow\,$}}}
\def\downto{{\raise 1pt \hbox{$\scriptstyle \,\searrow\,$}}}
\def\B{{\cal B}}

\def\M{{\cal M}}
\def\A{{\cal A}}
\def\N{{\cal N}}
\def\F{{\cal F}}

\def\reals{\mathbb{R}}
\def\ereals{\overline{\mathbb{R}}}
\def\uball{\mathbb{B}}

\def\comp{\raise 1pt \hbox{$\scriptstyle\circ$}}

\def\dom{\mathop{\rm dom}}

\def\ri{\mathop{\rm ri}}
\def\aff{\mathop{\rm aff}}

\def\inte{\mathop{\rm int}\nolimits}

\def\tos{\rightrightarrows}
\def\FF{(\F_t)_{t=0}^T}

\newcommand{\pk}{,\ldots ,}
\newcommand{\pke}{,\ldots}
\newcommand{\lk}{\left\{\,}
\newcommand{\rk}{\right\}}
\newcommand{\mk}{\;\big|\;}

\newtheorem{theorem}{Theorem}

\newtheorem{lemma}[theorem]{Lemma}
\newtheorem{corollary}[theorem]{Corollary}
\newtheorem{definition}[theorem]{Definition}
\newtheorem{example}[theorem]{Example}
\newtheorem{remark}[theorem]{Remark}
\newenvironment{proof}
{\begin{trivlist}\item[\, 
{\bf Proof.}]}{{\hfill $\square$}\end{trivlist}}

\title{Hedging of claims with physical delivery under convex transaction costs}
\author{Teemu Pennanen \and Irina Penner}

\begin{document}
\maketitle

\begin{abstract}
We study superhedging of contingent claims with physical delivery in a discrete-time market model with convex transaction costs. Our model extends Kabanov's currency market model by allowing for nonlinear illiquidity effects. We show that an appropriate generalization of Schachermayer's robust no arbitrage condition implies that the set of claims hedgeable with zero cost is closed in probability. Combined with classical techniques of convex analysis, the closedness yields a dual characterization of premium processes that are sufficient to superhedge a given claim process. We also extend the fundamental theorem of asset pricing for general conical models.
\end{abstract}

\section{Introduction}

This paper studies superhedging of contingent claims with physical delivery in markets with temporary illiquidity effects. Our market model is a generalization of the currency market model of Kabanov~\cite{kab99}. In Kabanov's model price dynamics and transaction costs are modeled implicitly by {\em solvency cones}, i.e.\ sets of portfolios which can be transformed into the zero portfolio by self-financing transactions at a given time and state. An essential difference between Kabanov's model and more traditional models of mathematical finance (including e.g.\ the transaction cost models of Jouini and Kallal~\cite{jk95a}, Cvitani{\'c} and Karatzas~\cite{ck96} and Kaval and Molchanov~\cite{km6}) is that Kabanov's model focuses on contingent claims with {\em physical delivery}, i.e.\ claims whose payouts are given in terms of portfolios of assets instead of a single reference asset like cash. Accordingly, notions of arbitrage as well as the corresponding dual variables are defined in terms of vector-valued processes; see e.g.\ \cite{krs3}, \cite{sch4}.

Astic and Touzi~\cite{at7} extended Kabanov's model by allowing general convex solvency regions in the case of finite probability spaces. Nonconical solvency regions allow for the modeling of temporary illiquidity effects where marginal trading costs may depend on the magnitude of a trade as e.g.\ in {\c{C}}etin, Jarrow and Protter~\cite{cjp4}, {\c{C}}etin and Rogers~\cite{cr7}, {\c{C}}etin, Soner and Touzi~\cite{cst7}, Rogers and Singh~\cite{rs7}, {\c{C}}etin and Rogers~\cite{cr7} or Pennanen~\cite{pen6}. These models cover nonlinear illiquidity effects but they assume that agents have no market power in the sense that their trades do not affect the costs of subsequent trades; see \cite{cr7} for further motivation of this assumption. Temporal illiquidity effects act essentially as nonlinear transaction costs. Moreover, most modern stock exchanges are organized so that the costs are convex with respect to transacted amounts; see \cite{pen6}.

This paper studies general convex solvency regions in general probability spaces in finite discrete time. Our main result gives a sufficient condition for the closedness of the set of contingent claims with physical delivery that can be hedged with zero investment. Classical separation arguments then yield dual characterizations of superhedging conditions much as in Pennanen~\cite{pen8b} in the case of claims with cash delivery. In the conical case, our sufficient condition coincides with the {\em robust no arbitrage condition} and the dual variables become {\em consistent price systems} in the sense of Schachermayer~\cite{sch4}. We also give a version of the ``fundamental theorem of asset pricing'' for general conical models. Even in the conical case our results improve on the existing ones since we do not assume polyhedrality of the solvency cones. 

The rest of this paper is organized as follows. Section~\ref{sec:mm} describes the market model and Section~\ref{sec:mr} gives the main result. Section~\ref{sec:sh} combines the main result with some classical techniques of convex analysis to derive dual characterizations of superhedging conditions. Section~\ref{sec:ftap} generalizes the fundamental theorem of asset pricing to the general conical case. The proof of the main result is contained in Section~\ref{sec:proof}.

\section{The market model}\label{sec:mm}

We consider a financial market in which $d$ securities can be traded over finite discrete time $t=0,\ldots,T$. The information evolves according to a filtration $\FF$ on a probability space $(\Omega,\F,P)$. 

For each $t$ and $\omega$ we denote by $C_t(\omega)\subset\reals^d$ the set of portfolios that are freely available in the market. We assume that for each $t$ the set-valued mapping $C_t:\Omega\tos\reals^d$ is $\F_t$-measurable in the sense that 
\[
C_t^{-1}(U):=\{\omega\in\Omega\,|\,C_t(\omega)\cap U\ne\emptyset\}\in\F_t
\]
for every open set $U\subset\reals^d$. 

\begin{definition}\label{def:mm}
A \emph{market model} is an $\FF$-adapted sequence $C=(C_t)_{t=0}^T$ of closed-valued mappings $C_t:\Omega\tos\reals^d$ with $\reals^d_-\subset C_t(\omega)$ for every $t$ and $\omega$. A market model $C$ is {\em convex, conical, polyhedral, \ldots} if $C_t(\omega)$ has the corresponding property for every $t$ and $\omega$.
\end{definition}

Traditionally, portfolios in financial market models have been defined in terms of a reference asset such as cash or some other num\'{e}raire; see Example~\ref{ex:ccp} below. This is natural when studying financial contracts with cash payments only. Treating all assets symmetrically as in Definition~\ref{def:mm} was initiated by Kabanov~\cite{kab99}.



\begin{example}[Currency markets with proportional transaction costs]\label{ex:kab}
If $(s_t)_{t=0}^T$ is an adapted price process with values in $\reals^d_+$ and $(\Lambda_t)_{t=0}^T$ an adapted $\reals_+^{d\times d}$-valued process of transaction cost coefficients, the {\em solvency regions} (in physical units) were defined in Kabanov~\cite{kab99} as
\[\hat{K}_t:=\{ x\in\reals^d\,|\, \exists a\in\reals_+^{d\times d}:\ s_t^ix^i+\sum_{j=1}^d(a^{ji}-(1+\lambda^{ij}_t)a^{ij})\ge0,\,1\le i\le d\}.\]
A portfolio $x$ belongs to the solvency region $\hat{K}_t$ iff, after some possible transfers $(a^{ji})_{1\le i,j\le d}$, it has only nonnegative components. Thus the solvency region describes the set of all portfolios with ``positive'' values. 

One can also define solvency regions directly in terms of bid-ask spreads as in Schachermayer~\cite{sch4}. If $(\Pi_t)_{t=0}^T$ is an adapted sequence of {\em bid-ask matrices}, then 
\[
\hat{K}_t = \{x\in\reals^d\,|\, \exists a\in\reals^{d\times d}_+:\ x^i + \sum_{j=1}^d (a^{ji}-\pi_t^{ij}a^{ij})\ge0,\,1\le i\le d\}.
\]
For each $\omega$ and $t$ the set $\hat K_t(\omega)$ is a polyhedral cone and
\[
C_t(\omega): = -\hat K_t(\omega)
\]
defines a conical market model in the sense of Definition~\ref{def:mm}.
\end{example}

\begin{example}[Illiquid markets with cash]\label{ex:ccp}
A \emph{convex cost process} is a sequence $S=(S_t)_{t=0}^T$ of extended real-valued functions on $\reals^d\times\Omega$ such that  for all $t$ the function $S_t$ is $\B(\reals^d)\otimes\F_t$-measurable and for each $\omega$ the function $S_t(\cdot,\omega)$ is lower semicontinuous, convex and vanishes at $0$. The quantity $S_t(x,\omega)$ denotes the cost (in cash) of buying a portfolio $x$ at time $t$ and scenario $\omega$; see \cite{pen6,pen8}. If $S$ is a convex cost process, then
\[
C_t(\omega) = \{x\in\reals^d\,|\, S_t(x,\omega)\le 0\},\quad t=0\pk T
\]
defines a convex market model. 

Models with convex cost processes include, in particular, classical frictionless markets (where $S_t(x,\omega)=s_t(\omega)\cdot x$ for an adapted $\reals^d$-valued price process $s$) as well as models with bid-ask spreads or proportional transaction costs as e.g.\ in Jouini and Kallal~\cite{jk95a}. Convex cost processes also allow for modeling of illiquidity effects as e.g.\ in {\c{C}}etin and Rogers~\cite{cr7} where $d=2$ and $S_t((y,x),\omega)=y+s_t(\omega)\varphi(x)$ for a strictly positive adapted process $(s_t)_{t=0}^T$ and an increasing convex function $\varphi:~\reals\rightarrow(-\infty,\infty]$. 

A convex cost process can be identified with the {\em liquidation function} $P_t$ in Astic and Touzi \cite{at7} through $S(x, \omega)=-P_t(-x,\omega)$. Convex cost processes are also related to the \emph{supply curve} introduced in {\c{C}}etin et al.~\cite{cjp4}. A supply curve $s_t(x,\omega)$ gives a price per unit when buying $x$ units of the risky asset so that the total cost is $S_t(x,\omega)=s_t(x,\omega)x$. Instead of convexity of $S$, \cite{cjp4} assumed that the supply curve is smooth in $x$; see also Example 2.2 in \cite{at7}. For more examples of convex cost processes and their properties we refer to \cite{pen8,pen8b}.
\end{example}

\begin{example}[Currency markets with illiquidity costs]
In order to model nonproportional illiquidity effects in a currency market model as in Example~\ref{ex:kab} one can replace a bid-ask matrix $(\Pi_t)_{t=0}^T$ by a matrix of convex cost processes $S^{ij}=(S^{ij}_t)_{t=0}^T$ $(1\le i,j\le d)$ on $\reals_+$. Here $S^{ij}(x,\omega)$ denotes the number of units of asset $i$ for which one can buy $x$ units of asset $j$. In a market with proportional transaction costs we simply have $S^{ij}(x,\omega)=\pi^{ij}(\omega)x$. If $(S^{ij}_t)_{t=0}^T$, $i,j=1,\ldots,d$ are convex cost processes on $\reals_+$ in the sense of Example~\ref{ex:ccp}, then
\[
C_t(\omega) = \{x\in\reals^d\,|\, \exists a\in\reals^{d\times d}_+:\ x^i \le \sum_{j=1}^d(a^{ji}-S_t^{ij}(a^{ij},\omega)),\,1\le i\le d\}, 
\]
for $t=0\pk T$ defines a convex market model.
\end{example}

\section{The main result}\label{sec:mr}


An $\reals^d$-valued adapted process $x=(x_t)_{t=0}^T$ is a {\em self-financing portfolio process} in a market model $C=(C_t)_{t=0}^T$ if
\[
\Delta x_t:=x_t-x_{t-1}\in C_t\quad P\text{-a.s.}
\]
for every $t=0,\ldots,T$, i.e.\ the increments $\Delta x_t$ are freely available in the market. Here and in what follows, we always define $x_{-1}=0$. 

We say that a market model $C$ has the {\em no-arbitrage} property if
\begin{equation}\label{na}
A_T(C)\cap L^0_+=\{0\},
\end{equation}
where $A_T(C)$ denotes the convex set in $L^0$ formed by final values $x_T$ of all self-financing portfolio processes $x=(x_t)_{t=0}^T$. Since $x_T=\sum_{t=0}^T(x_t-x_{t-1})$, we have the expression
\[
A_T(C)=L^0(C_0,\F_0)+\ldots+L^0(C_T,\F_T),
\]
where $L^0(C_t,\F_t)$ denotes the set of all $\F_t$-measurable selectors of $C_t$, i.e.\ the set of all $\F_t$-measurable random vectors $x$ such that $x\in C_t$ almost surely. Condition \eqref{na} was introduced in Kabanov and Stricker~\cite{ks1b} under the name ``weak no-arbitrage property'' in a formally different way. But it is equivalent to \eqref{na} if $C_t$ contains $\reals_-^d$ for all $t$ as noted in Lemma 3.5. of Kabanov~\cite{kab1}.

In classical market models the no-arbitrage condition implies the closedness of the set of contingent claims with cash-delivery that can be superhedged at zero cost, a result which is of vital importance in deriving dual characterizations of superhedging and absence of arbitrage. However, as shown in Schachermayer~\cite{sch4}, in a market with proportional transaction costs as in Example~\ref{ex:kab}, the no-arbitrage property \eqref{na} does not, in general, imply the closedness of the set $A_T(C)$. Schachermayer~\cite{sch4} also showed, in the case of the conical model of Example~\ref{ex:kab}, that $A_T(C)$ is closed is probability if $C$ satisfies the \emph{robust no-arbitrage} condition which can be defined in the general conical case as follows.

Given a market model $C$ let $C^0_t(\omega)$ be the largest linear subspace contained in $C_t(\omega)$. Using the terminology of Kabanov, R{\'a}sonyi and Stricker~\cite{krs3} we say that $C$ is \emph{dominated} by another market model $\tilde C$ if 
\[
C_t\setminus C^0_t\subset\ri\tilde C_t\quad\text{for all}\quad t=0\pk T.
\]
\begin{definition}
A conical market model has the {\em robust no arbitrage property} if it is dominated by another conical market model that has the no arbitrage property. 
\end{definition}

When moving to general convex market models it is not immediately clear how the condition of robust no-arbitrage should be extended in order to have the closedness of $A_T(C)$. Indeed, in general convex models even the traditional notion of arbitrage has two natural extensions, one being the possibility of making something out of nothing the other one being the possibility of making arbitrarily much out of nothing; see \cite{pen8}. These correspond to the notions of the {\em tangent cone} and the {\em recession cone} from convex analysis. It turns out to be the latter one which is more relevant for closedness of $A_T(C)$; see \cite{pen8b} for the case of claims with cash-delivery. This is, in fact, suggested already by classical closedness criteria in convex analysis; see \cite[Chapter~8]{roc70a}.

Given a convex market model $C$, let
\[
C^\infty_t(\omega) = \{x\in\reals^d\,|\, C_t(\omega)+\alpha x\subset C_t(\omega)\ \forall\alpha >0\}.
\]
This is a closed convex cone known as the {\em recession cone} of $C_t(\omega)$; see \cite[Chapter~8]{roc70a}. The recession cone describes the asymptotic behavior of a convex set infinitely far from the origin. Since $C_t(\omega)$ is a closed convex set containing $\reals^d_-$ we have, by \cite[Theorem~8.1, Theorem~8.2, Corollary~8.3.2, Theorem~8.3]{roc70a}, that $C^\infty$ is a closed convex cone containing $\reals^d_-$ and 
\[
C^\infty_t(\omega) = \bigcap_{\alpha>0}\alpha C_t(\omega)
\]
and
\begin{equation}\label{recession}
C^\infty_t(\omega) = \{x\in\reals^d\,|\,\exists x^n\in C_t(\omega),\ \alpha^n\downto 0,\ \text{with}\ \alpha^nx^n\to x\}.
\end{equation}
By \cite[Exercise~14.21]{rw98}, the set-valued mappings $\omega\mapsto C_t^\infty(\omega)$ are $\F_t$-measurable so they define a convex conical market model in the sense of Definition~\ref{def:mm}. 

\begin{definition}
A convex market model has the {\em robust no scalable arbitrage property} if $C^\infty$ has the robust no arbitrage property. 
\end{definition}

The term ``scalable arbitrage'' refers to arbitrage opportunities that may be scaled by arbitrarily large positive numbers to yield arbitrarily ``large'' arbitrage opportunities; see \cite{pen8}. Such scalable arbitrage opportunities can be related to the market model $C^\infty$ much as in \cite[Proposition~17]{pen8}.

We are now ready to state our main result the proof of which can be found in the last section.

\begin{theorem}\label{thm:cl}
If $C$ is a convex market model with the robust no scalable arbitrage property then $A_T(C)$ is closed in probability.
\end{theorem}

\begin{remark}\label{rem:arb}
If $C$ is conical, we have $C^\infty_t(\omega)=C_t(\omega)$ and Theorem~\ref{thm:cl} coincides with \cite[Lemma~2]{krs3} which extends \cite[Theorem~2.1]{sch4} to general conical models.
\end{remark}

\begin{remark}
For general convex models $C^\infty_t(\omega)\subset C_t(\omega)$ and the condition in Theorem~\ref{thm:cl} may be satisfied even if $C$  fails the no-arbitrage condition. Consider for example a deterministic model where $\Omega$ is a singleton and $C_t(\omega)=\reals_-^d+\uball$ for every $t$. Here $\uball$ denotes the unit ball of $\reals^d$. We get,
\[
A_T(C)=L^0(C_0,\F_0)+\ldots+L^0(C_T,\F_T)=\reals_-^d+(T+1)\uball
\]
so $C$ does not have the no-arbitrage property. On the other hand, $C_t^\infty(\omega)=\reals_-^d$ is dominated by $\tilde C_t(\omega)=\{x\in\reals^d\,|\,\sum_{i=1}^dx^i\le 0\}$ which does have the no-arbitrage property. Indeed, $A_T(\tilde C)=\{x\in\reals^d\,|\,\sum_{i=1}^dx^i\le 0\}$ so $A_T(\tilde C)\cap L^0_+=\{0\}$.
\end{remark}

\section{Superhedging}\label{sec:sh}

A contingent claim with {\em physical delivery} is a security that, at some future time, gives its owner a random portfolio of securities (instead of a single security like in the case of cash-delivery). A {\em contingent claim process with physical delivery} $c=(c_t)_{t=0}^T$ is a security that, at each time $t=0,\ldots,T$, gives its owner an $\F_t$-measurable random portfolio $c_t\in\reals^d$. The set of $\reals^d$-valued adapted process will be denoted by $\A$.

Given a market model $C$, we will say that a process $p\in\A$ is a {\em superhedging premium process} for a claim process $c\in\A$ if there is a portfolio process $x\in\A$ such that $x_T=0$ and 
\[
x_t-x_{t-1}+c_t-p_t\in C_t
\]
almost surely for every $t=0,\ldots,T$. The requirement that $x_T=0$ means that everything is liquidated at the terminal date. One could relax this condition to $x_T\ge 0$ but since $\reals_-^d\subset C_T$ it would amount to the same thing.

We use claim and premium processes (rather than claims and premiums (prices) with pay-outs only at the end and the beginning) in the present paper mainly for mathematical convenience. However, when moving to market models with portfolio constraints it is essential to distinguish between payments at different points in time, and then claim processes become the natural object of study; see \cite{pen8,pen8b}. Claim and premium processes are common in various insurance applications where payments are made e.g.\ annually.

The superhedging condition can be written as 
\[
c-p\in A(C),
\]
where
\[
A(C) = \{c\in\A\,|\,\exists x\in\A:\ x_t-x_{t-1}+c_t\in C_t,\ x_T=0\},
\]
is the set of all claim processes with physical delivery that can be superhedged without any investment. It is easily checked that $A(C)$ is convex (conical) when $C$ is convex (conical).

\begin{lemma}\label{lem:ata}
Let $C$ be a convex market model.
\begin{enumerate}
\item
The sets $A(C)$ and $A_T(C)$ are related through
\begin{align*}
A_T(C) &= \{c_T\,|\,(0,\ldots,0,c_T)\in A(C)\},\\
A(C) &= \{(c_0,\ldots,c_T)\,|\, \sum_{t=0}^Tc_t\in A_T(C)\}.
\end{align*}
\item
The set $A(C)$ is closed in probability if and only if $A_T(C)$ is closed in probability.
\item
$C$ has the no-arbitrage property if and only if $A(C)\cap\A_+=\{0\}$.
\end{enumerate}
\end{lemma}

\begin{proof}
It suffices to prove the first part since the other two follow from that. The first equation is immediate. As to the second, we have $c\in A(C)$ iff there is an $x\in\A$ such that $x_T=0$ and $x_t-x_{t-1}+c_t\in C_t$. Defining $\tilde x_{-1}=0$ and $\tilde x_t=\tilde x_{t-1}+x_t-x_{t-1}+c_t$ we get that $\tilde x\in\A$ is self-financing and $\tilde x_T=\sum_{t=0}^Tc_t$. This just means that $\sum_{t=0}^Tc_t\in A_T(C)$.
\end{proof}

\begin{example}\label{ex:static}
The classical case of a single premium payment at the beginning and single claim payment at maturity corresponds to $p=(p_0,0,\ldots,0)$ and $c=(0,\ldots,0,c_T)$. In this case, Lemma~\ref{lem:ata} says that $p$ is a superhedging premium for $c$ if and only if
\begin{equation}\label{hs}
c_T-p_0\in A_T(C).
\end{equation}
\end{example}

Dual characterizations of the set of all initial endowments satisfying condition \eqref{hs} have been given in the conical case as in Example \ref{ex:kab} in \cite{ks1b}, \cite{dkv2}, \cite{krs2}, \cite{ipen1} and \cite{sch4}. In conical models the superhedging endowments are characterized in terms of the same dual elements that characterize the no-arbitrage condition. When moving to nonconical models, a larger class of dual variables is needed in order to capture the structure of the sets; see \cite{pen8b} for the case of claims with cash delivery. 

By Lemma~\ref{lem:ata}, the set $A(C)$ is closed if and only if $A_T(C)$ is closed. Combining Theorem~\ref{thm:cl} with classical techniques of convex analysis, we can derive dual characterizations of superhedging premium process in terms of martingales much as in \cite{pen8b} in the case of claims with cash delivery. Consider the Banach space 
\[
\A^1:=\lk x\in\A\mk x_t\in L^1(P)\quad\mbox{for all}\;t=0\pk T\rk
\]
and its dual
\[
\A^\infty:=\lk x\in\A\mk x_t\in L^\infty(P)\quad\mbox{for all}\;t=0\pk T\rk.
\]
One can then use the classical bipolar theorem to characterize the superhedging condition in terms of the {\em support function} $\sigma_{A^1(C)}:\A^\infty\to\ereals$ of the set $A^1(C):=A(C)\cap\A^1$ of integrable claim processes. The support function is given by
\[
\sigma_{A^1(C)}(y) = \sup_{c\in A^1(C)}E\sum_{t=0}^Tc_t\cdot y_t.
\]
The lemma below expresses $\sigma_{A^1(C)}$ in terms of the support functions of the random sets $C_t(\omega)$
\[
\sigma_{C_t(\omega)}(y) = \sup_{x\in C_t(\omega)}x\cdot y.
\]
By \cite[Example~14.51]{rw98} the function $\sigma_{C_t}:\Omega\times\reals^d\to\ereals$ is an $\F_t$-measurable normal integrand (see \cite[Definition~14.27]{rw98}). This implies in particular that $\sigma_{C_t(\omega)}(y_t)$ is an  $\F_t$-measurable function whenever $y_t$ is $\F_t$-measurable. The following corresponds to \cite[Lemma~28]{pen8}.

\begin{lemma}\label{lem:support} 
Let $C$ be a convex market model and let $y\in\A^\infty$. Then
\[
\sigma_{A^1(C)}(y) = 
\begin{cases}
E\sum_{t=0}^T\sigma_{C_t}(y_t) & \text{if $y$ is a nonnegative martingale},\\
+\infty & \text{otherwise}.
\end{cases}
\]
\end{lemma}
\begin{proof}
In the following, we define the expectation of an arbitrary random variable $\varphi$ by setting it equal to $-\infty$ if the negative part of $\varphi$ is not integrable (the remaining cases being defined unambiguously as real numbers or as $+\infty$).

On one hand,
\begin{align*}
\sigma_{A^1(C)}(y) &= \sup\{E\sum_{t=0}^Tc_t\cdot y_t\,|\, x\in\A,\ c\in\A^1:\ x_t-x_{t-1}+c_t\in C_t\}\\
&\le \sup\{E\sum_{t=0}^Tc_t\cdot y_t\,|\, x,c\in\A:\ x_t-x_{t-1}+c_t\in C_t\}\\
&=\sup\{E\sum_{t=0}^T(w_t-x_t+x_{t-1})\cdot y_t\,|\, x,w\in\A:\ w_t\in C_t\}\\
&\le\sup\{E\sum_{t=0}^T[\sigma_{C_t}(y_t)+(x_{t-1}-x_t)\cdot y_t]\,|\, x\in\A\}\\
&= E\sum_{t=0}^T\sigma_{C_t}(y_t) + \sup_{x\in\A}E\sum_{t=0}^T(x_{t-1}-x_t)\cdot y_t\\
&=E\sum_{t=0}^T\sigma_{C_t}(y_t) + \sup_{x\in\A}E\sum_{t=0}^{T-1}x_t\cdot(y_{t+1}-y_t),
\end{align*}
where the last term vanishes if $y$ is a martingale and equals $+\infty$ otherwise. On the other hand,
\begin{align*}
\sigma_{A^1(C)}(y) &= \sup\{E\sum_{t=0}^Tc_t\cdot y_t\,|\, x\in\A,\ c\in A^1:\ x_t-x_{t-1}+c_t\in C_t\}\\
&\ge\sup\{E\sum_{t=0}^Tc_t\cdot y_t\,|\, x,c\in\A^1:\ x_t-x_{t-1}+c_t\in C_t\}\\
&=\sup\{E\sum_{t=0}^T(w_t-x_t+x_{t-1})\cdot y_t\,|\, x,w\in\A^1:\ w_t\in C_t\}\\
&=\sup\{E\sum_{t=0}^Tw_t\cdot y_t\,|\, w\in\A^1:\ w_t\in C_t\} + \sup_{x\in\A^1} E\sum_{t=0}^T(x_{t-1}-x_t)\cdot y_t\\
&=\sum_{t=0}^T\sup\{Ew_t\cdot y_t\,|\, w_t\in L^1(\F_t):\ w_t\in C_t\} + \sup_{x\in\A^1}E\sum_{t=0}^{T-1}x_t\cdot(y_{t+1}-y_t). 
\end{align*}
Here again the last term vanishes if $y$ is a martingale and equals $+\infty$ otherwise. By \cite[Theorem~14.60]{rw98}, 
\[
\sup\{Ew_t\cdot y_t\,|\, w_t\in L^1(\F_t):\ w_t\in C_t\} = E\sigma_{C_t}(y_t),
\]
which proves that
\[
\sigma_{A^1(C)}(y) = 
\begin{cases}
E\sum_{t=0}^T\sigma_{C_t}(y_t) & \text{if $y$ is a martingale},\\
+\infty & \text{otherwise}.
\end{cases}
\]
Moreover, since $\reals_-^d\subseteq C_t$ for all $t$, we have $\sigma_{C_t}(y_t)=\infty$ on the set $\{y_t\notin\reals_+^d\}$. This completes the proof.
\end{proof}

\begin{theorem}\label{thm:hedging}
Assume that the convex market model $C$ has the robust no scalable arbitrage property and let $c,p\in\A$ be such that $c-p\in\A^1$. Then the following are equivalent
\begin{itemize}
\item[(i)]
$p$ is a superhedging premium process for $c$.
\item[(ii)]
$E\sum_{t=0}^T(c_t-p_t)\cdot y_t \le 1$ for every bounded nonnegative martingale $y=(y_t)_{t=0}^T$ such that $E\sum_{t=0}^T\sigma_{C_t}(y_t)\le 1$.
\item[(iii)]
$E\sum_{t=0}^T(c_t-p_t)\cdot y_t \le E\sum_{t=0}^T\sigma_{C_t}(y_t)$ for every bounded nonnegative martingale $y=(y_t)_{t=0}^T$.
\end{itemize}
\end{theorem}

\begin{proof}
By definition, $p$ is a super hedging premium for $c$ if and only if $c-p\in A(C)$. Since $c-p\in\A^1$, this can be written as $c-p\in A^1(C)$, where $A^1(C)=A(C)\cap\A^1$ is a closed subset of $\A^1$, by Theorem~\ref{thm:cl} and Lemma~\ref{lem:ata}. Since $A^1(C)$ is also a convex set containing the origin, the bipolar theorem (see e.g.\ \cite[Theorem~5.91]{ab99}) implies that $c-p\in A^1(C)$ iff 
\[
E\sum_{t=0}^T(c_t-p_t)\cdot y_t\le 1\quad\forall y\in A^1(C)^\circ,
\]
where 
\[
A^1(C)^\circ = \{y\in\A^\infty\,|\, \sigma_{A^1(C)}(y)\le 1\}.
\]
Thus the equivalence of (i) and (ii) and the proof of (i) $\Rightarrow$ (iii) follow from Lemma \ref{lem:support}. And obviously (iii) implies (ii).
\end{proof}

If $C$ is conical, we have
\begin{equation}\label{consup}
\sigma_{C_t(\omega)}(y) = 
\begin{cases}
0 & \text{if}  \; y\in C_t^*(\omega),\\
+\infty & \text{otherwise}
\end{cases}
\end{equation}
where the closed convex cone
\[
C_t^*(\omega):=\{y\in\reals^d\,|\, x\cdot y\le 0\ \forall x\in C_t(\omega)\}
\]
is known as the {\em polar} of $C_t(\omega)$; see \cite{roc70a}. We then have $E\sum_{t=0}^T\sigma_{C_t}(y_t)<\infty$ for a $y\in\A^\infty$ if and only if $y_t\in C_t^*$ almost surely for all $t$. Thus, in the conical case, superhedging premiums can be characterized in terms the following dual elements introduced in Kabanov~\cite{kab99}.

\begin{definition}\label{cps}
An adapted $\reals^d$-valued non-zero process $y=(y_t)_{t=0}^T$ is called a {\em consistent (resp. strictly consistent) price system for a conical model} $C$ if $y$ is a martingale such that $y_t\in  C^*_t$ (resp. $y$ has strictly positive components and $y_t\in \ri C^*_t$) almost surely for all $t=0\pk T$.
\end{definition}

Note that the condition $y_t\in \ri C^*_t$ does not imply the strict positivity of $y_t$. Strict positivity of a strictly consistent price system $y$ is not explicitly required in \cite{krs2} or in \cite{sch4} but it is used and obtained in the proofs of their main results. 

Applying Theorem \ref{thm:hedging} in the conical case and making use of \eqref{hs} and \eqref{consup}, we obtain the following.

\begin{corollary}\label{cor:hedgconical}
Assume that $C$ is a conical market model with the robust no arbitrage property. Assume further that $\F_0$ is a trivial $\sigma$-algebra. Let $c_T\in L^1(P)$ and $p_0\in\reals$. Then the following are equivalent.
\begin{itemize}
\item[(i)]
$p=(p_0,0,\ldots,0)$ is a superhedging premium for $c=(0,\ldots,0,c_T)$.
\item[(ii)]
$E(c_T\cdot y_T)\le p_0\cdot y_0$ for every bounded consistent price system $y=(y_t)_{t=0}^T$.
\end{itemize}
\end{corollary}

Similar characterizations were given in \cite{ks1b}, \cite{dkv2}, \cite{krs2}, \cite{ipen1} and \cite{sch4} under less restrictive integrability conditions on $c_T$. In our case the integrability condition in Theorem \ref{thm:hedging} can be relaxed in the following way. If the process $c-p$ is not $P$-integrable, we can always find a probability measure $\tilde{P}\approx P$ with bounded density such that $c_t-p_t\in L^1(\tilde{P})$ for all $t$, e.g. 
\[\frac{d\tilde{P}}{dP}:=\frac{a}{1+\sum_{t=0}^T|c_t-p_t|}, 
\]
where $a$ is a normalizing constant. Then Theorem \ref{thm:hedging} holds with $\tilde{P}$ instead of $P$ and we obtain the following corollary. 

\begin{corollary}\label{cor:hedging}
Assume that $C$ is a conical market model with the robust no arbitrage property and let $c,p\in\A$. Let further $\tilde{P}\approx P$ be a probability measure with bounded density process $z=(z_t)_{t=0}^T$ such that $c_t-p_t\in L^1(\tilde{P})$ for all $t$. Then the following are equivalent
\begin{itemize}
\item[(i)]
$p$ is a superhedging premium process for $c$.
\item[(ii)]
$E\sum_{t=0}^T(c_t-p_t)\cdot y_t \le 1$ for every bounded nonnegative $P$-martingale $y=(y_t)_{t=0}^T$ such that $E\sum_{t=0}^T\sigma_{C_t}(y_t)\le 1$ and such that $y_t/z_t:=(y_t^1/z_t\pk y_t^d/z_t)$ is almost surely bounded for all $t$.
\end{itemize}
\end{corollary}

\begin{proof}
Theorem \ref{thm:hedging} applied with $\tilde{P}$ instead of $P$ yields the equivalence of the following \begin{itemize}
\item[(i)]
$p$ is a superhedging premium process for $c$.
\item[(ii)]
$E\sum_{t=0}^T(c_t-p_t)\cdot z_t\tilde{y}_t \le 1$ for every nonnegative bounded $\tilde{P}$-martingale $\tilde{y}=(\tilde{y}_t)_{t=0}^T$ such that $E\sum_{t=0}^Tz_t\sigma_{C_t}(\tilde{y}_t)\le 1$.
\end{itemize}
Note further that $z\sigma_{C_t}(y)=\sigma_{C_t}(zy)$ for all $y\in\reals^d, z\in\reals^+$, and that $\tilde{y}$ is a bounded $\tilde{P}$-martingale iff $z\tilde{y}$ is a bounded $P$-martingale. Thus (ii) holds for all nonnegative bounded $P$-martingales $y$ such that $\frac{1}{z}y$ is almost surely bounded.
\end{proof}

\section{Fundamental theorem of asset pricing}\label{sec:ftap}

Fundamental theorem of asset pricing describes absence of arbitrage by existence of certain pricing functionals. In the classical frictionless model with a cash account these pricing functionals can be identified with equivalent martingale measures. In conical models the robust no-arbitrage condition is equivalent to existence of strictly consistent price systems. This result was proved in \cite{ks1b}, \cite{krs2}, \cite{ipen1}, \cite{sch4}, \cite{krs3} for conical {\em polyhedral} market models, i.e.\ where each $C_t(\omega)$ is a cone spanned by a finite number of vectors. Non-polyhedral conical models are considered in \cite{ras7} and in \cite{rok7} under the assumption of efficient friction, i.e.\ the cones $C_t(\omega)$ are assumed to be pointed. 

In this section we derive a version of the fundamental theorem of asset pricing for general conical models. This is achieved through the following lemma which allows the application of strict separation arguments much as in \cite{sch4}. It is similar to Proposition A.5 in \cite{sch4} but does not rely on polyhedrality of $C$.

\begin{lemma}\label{lem:dom}
If $K$ is a conical market model that has the robust no arbitrage property, then it is dominated by another conical market model $\hat K$ which still has the robust no arbitrage property.
\end{lemma}

\begin{proof}
If $K$ has the robust no arbitrage property, there is an arbitrage free model $\tilde K$ such that $K_t\setminus K_t^0\subset\inte\tilde K_t$, or equivalently, $\tilde K_t^*\setminus\{0\}\subset\ri K_t^*$. Let $x_t$ be an $\F_t$-measurable vector such that $x_t\in\inte\tilde K_t$. We can write $K^*_t(\omega)$ and $\tilde K^*_t(\omega)$ as 
\[
K_t^*(\omega) = \bigcup_{\alpha\ge 0}\alpha G_t(\omega)\quad\text{and}\quad\tilde K_t^*(\omega) = \bigcup_{\alpha\ge 0}\alpha\tilde G_t(\omega),
\]
where $G_t(\omega)=\{v\in\reals^d\,|\, v\in K_t^*(\omega),\ x_t(\omega)\cdot v=-1\}$ and $\tilde G_t(\omega)=\{v\in\reals^d\,|\, v\in\tilde K^*_t(\omega),\ x_t(\omega)\cdot v=-1\}$. Since $\tilde K_t^*\setminus\{0\}\subset\ri K_t^*$ we have $\tilde G_t(\omega)\subset\ri G_t(\omega)$. It suffices to show that there is an $\F_t$-measurable closed convex set $\hat G_t(\omega)$ such that $\tilde G_t(\omega)\subset\ri\hat G_t(\omega)\subset\ri G_t(\omega)$. Indeed, the polar of the cone 
\[
\hat K_t^*(\omega) = \bigcup_{\alpha\ge 0}\alpha\hat G_t(\omega)
\]
will then dominate $K$ and be dominated by $\tilde K$.

Since $\tilde K_t(\omega)$ has nonempty interior, the set $\tilde G_t(\omega)$ is compact. This implies that there exists an $\varepsilon_t(\omega)>0$ such that $[\tilde G_t(\omega)+\varepsilon_t(\omega)\uball]\cap\aff G_t(\omega)\subset G_t(\omega)$ so the set $\hat G_t(\omega):=[\tilde G_t(\omega)+\frac{1}{2}\varepsilon_t(\omega)\uball]\cap\aff G_t(\omega)$ has the desired properties. 
\end{proof}

Equipped with Theorem~\ref{thm:cl} and Lemma~\ref{lem:dom} it is easy to extend the proof of \cite[Theorem~1.7]{sch4} to get the following.

\begin{theorem}\label{thm:ftap1}
A conical market model $K$ has the robust no arbitrage property if and only if there exists a strictly consistent price system $y$ for $K$. Moreover, the price system $y$ can be chosen bounded.
\end{theorem}

\begin{proof}
Lemma~\ref{lem:dom} implies the existence of another conical market model $\hat K$ such that $\hat K_t^*\setminus\{0\}\subset\ri (K_t)^*$. By Theorem~\ref{thm:cl} and Lemma~\ref{lem:ata}, the set $A(\hat K)$ is closed with respect to convergence in measure. Thus $A^1(\hat K)=A(\hat K)\cap\A^1$ is a convex cone in $\A^1$ which is closed in the norm topology. Moreover, the no-arbitrage property implies $A^1(\hat K)\cap\A^1_+=\{0\}$, where $\A^1_+$ denotes the $\reals^d_+$-valued processes in $\A^1$.

The Banach space $\A^1$ can be identified with $L^1(\widehat{\Omega}, \widehat{\F}, \mu)$, where $(\widehat{\Omega}, \widehat{\F})$ denotes the product of the spaces $(\Omega, \F_t), t=0\pk T$ and $\mu$ is a sigma-finite measure defined by the product  $P\times\nu$ for the counting measure $\nu$ on $\{0\pk T\}$. Thus by \cite[Lemma 2, Theorem~1]{jns5} the Kreps-Yan theorem holds true on $\A^1$, i.e.\ there exists a $y\in\A^\infty$ such that 
\begin{equation}\label{first}
E\sum_{t=0}^Ty_t\cdot c_t\le 0\quad \forall c\in A^1(\hat K)
\end{equation}
and 
\begin{equation}\label{second}
E\sum_{t=0}^Ty_t\cdot c_t > 0\quad \forall c\in\A^1_+\setminus\{0\}.
\end{equation}
Condition \eqref{first} can be written as $\sigma_{A^1(\hat K)}(y)\le 0$ which, by Lemma \ref{lem:support}, means that $y$ is a nonnegative martingale with $\sigma_{\hat K_t(\omega)}(y_t(\omega))\le 0$ almost surely for every $t$. Since $\hat K_t(\omega)$ is a cone, we have $y_t(\omega)\in \hat K_t(\omega)^*$. Condition \eqref{second} means that $y$ is componentwise strictly positive so $y_t\in\hat K_t^*\setminus\{0\}\subset\ri K_t^*$ and $y$ is a strictly consistent price system for $K$.

To prove the converse let $y$ be a strictly consistent price system for $K$ and define a conical model
\[
\tilde K_t:=\lk x\in\reals^d\mk y_t\cdot x\le 0\rk,\quad t=0\pk T.
\]
Since $y\in\ri K$ the model $\tilde K$ dominates $K$. It suffices to show that $\tilde K$ satisfies the no arbitrage condition. To this end let $c\in A(\tilde K)\cap\A_+$ and let $x\in\A$ be a self financing strategy that hedges $c$. Then we have
\[
y_t\cdot (x_t-x_{t-1}+c_t)\le0\quad\mbox{for all}\quad t=0\pk T
\] and hence
\begin{equation*}
0\le\sum_{t=0}^Ty_t\cdot c_t\le-\sum_{t=0}^Ty_t\cdot (x_t-x_{t-1})=\-\sum_{t=0}^{T}x_{t-1}\cdot (y_{t}-y_{t-1}).
\end{equation*}
Since $y$ is a martingale and $x$ is adapted, the process
\[
M_s:=\sum_{t=0}^{s}x_{t-1}\cdot (y_{t}-y_{t-1}),\quad s=0\pk T
\]
is a local martingale by Theorem 1 in Jacod and Shrirjaev \cite{js98}. Moreover, $M$ is a martingale by \cite[Theorem 2]{js98} since $M_T\ge0$ almost surely. Since $M_0=0$ we obtain
\[
E\left[\sum_{t=0}^Ty_t\cdot c_t\right]=0
\]
and hence $\sum_{t=0}^Ty_t\cdot c_t=0$ almost surely. Since $y$ has strictly positive components and $c\in\A_+$ this implies $c_t=0\;P$-a.s. for all $t$. Thus the no arbitrage condition holds for $\tilde K$.
\end{proof}

Using Theorem~\ref{thm:ftap1} we can restate Corollary~\ref{cor:hedgconical} in terms of strictly consistent price systems. 

\begin{corollary}
Assume that $C$ is a conical market model with the robust no arbitrage property. Assume further that $\F_0$ is a trivial $\sigma$-algebra. Let $c_T\in L^1(P)$ and $p_0\in\reals$. Then the following are equivalent.
\begin{itemize}
\item[(i)]
$p=(p_0,0,\ldots,0)$ is a superhedging premium for $c=(0,\ldots,0,c_T)$.
\item[(ii)]
$E(c_T\cdot y_T)\le p_0\cdot y_0$ for every bounded strictly consistent price system $y=(y_t)_{t=0}^T$.
\end{itemize}
\end{corollary}

\begin{proof}
By Corollary~\ref{cor:hedgconical}, it suffices to show that (ii) implies (i). Theorem \ref{thm:ftap1} says that there exists a strictly consistent price system $y^*$ for $C$. Then for $\varepsilon\in(0,1]$ and for any consistent price system $y$ the process $y^\varepsilon:= \varepsilon y^*+(1-\varepsilon)y$ defines a strictly consistent price system and for $\varepsilon$ small enough  $E(c_T\cdot y^\varepsilon_T)> p_0\cdot y^\varepsilon _0$ if $E(c_T\cdot y_T)> p_0\cdot y_0$.
\end{proof}

For general convex models, Theorem~\ref{thm:ftap1} can be stated in the following form.

\begin{corollary}
A convex market model $C$ satisfies the robust no scalable arbitrage property if and only if there exists a strictly positive martingale $y$ such that $y_t\in\ri\dom\sigma_{C_t}$ almost surely.
\end{corollary}

\begin{proof}
Applying Theorem~\ref{thm:ftap1} to $C^\infty$, we get that $C^\infty$ has the robust no-arbitrage property iff there exists a strictly consistent price system for $C^\infty$. The claim follows by noting that $\ri(C_t^\infty)^*=\ri\dom\sigma_{C_t}$, by \cite[Theorem~6.3]{roc70a}, the bipolar theorem and \cite[Corollary~14.2.1]{roc70a}.
\end{proof}

As noted in Remark~\ref{rem:arb}, arbitrage opportunities may very well exist under the conditions of Theorem~\ref{thm:cl}. They do not interfere with the hedging description in Theorem \ref{thm:hedging}. In Section 5 of \cite{cr7} it was also noted that arbitrage opportunities and optimal portfolios could coexist in an illiquid market. If one wants to exclude all arbitrage opportunities, one can formulate more restrictive no-arbitrage conditions as e.g.\ in \cite{at7} or \cite{pen8}.

\section{Proof of Theorem~\ref{thm:cl}}\label{sec:proof}

The proof of Theorem~\ref{thm:cl} requires some preparation. First we will use projection techniques similar to those in \cite{sch4} to extract self-financing portfolio processes ending with $0$ in the model $C^\infty$. For a convex market model $C$ we consider the set
\[ 
\N(C):=\lk x\in\A \mk \Delta x_t\in C_t\:\, P\text{-a.s.}\:\forall\,t=0\pk T,\;\, x_T=0\rk.
\]
The next lemma is a version of Lemma 5 in \cite{krs3}, see also Lemma 2.6 in \cite{sch4}. 

\begin{lemma}\label{lem:zero}
If $C$ is a convex market model that has the robust no scalable arbitrage property then $\N(C^\infty)=\N(C^0)$.
\end{lemma}

\begin{proof}
We have $\N(C^\infty)\supset\N(C^0)$ simply because $C^\infty_t\supset C^0_t$ almost surely for every $t$, so it suffices to prove the reverse. Let $x\in\N(C^\infty)$ and assume that $\Delta x_t\in C^\infty_t\setminus C_t^0$ on some set with positive probability for some $t=0,\ldots,T$. This contradicts the robust no scalable arbitrage property. Indeed, if $\tilde C$ is a market model such that $C^\infty_t\setminus C_t^0\subset\inte\tilde C_t$, we have $\Delta x_t\in\inte\tilde C_t$ on a nontrivial set, and then for any $e\in\reals^d_+\setminus\{0\}$
\[
\varepsilon_t(\omega)=\sup\{\varepsilon\in\reals\,|\, \Delta x_t(\omega) + \varepsilon e\in\tilde C_t(\omega)\}
\]
defines an $\F_t$-measurable $\reals^d_+$-valued random variable (see e.g.\ \cite[Theorem~14.37]{rw98}) which does not vanish almost surely. By Lemma~\ref{lem:ata}, this would mean that $\tilde C$ violates the no-arbitrage condition.
\end{proof}

For each $t\in\{0\pk T\}$ we denote by $\N_t$ the set of all $\F_t$-measurable random vectors that may be extended to a portfolio in $\N(C^0)$, i.e.
\[
\N_t:=\lk x_t\in L^0(\reals^d,\F_t)\mk \exists \,x_{t+1}\pk x_{T}\,\text{s.t.}\,(0\pk0,x_t\pk x_T)\in\N(C^0)\rk.
\]

\begin{lemma}\label{lem:proj}
Let $C$ be a convex market model. Then for each $t\in\{0\pk T\}$ there is an $\F_t$-measurable map $N_t:\Omega\tos\reals^d$ whose values are linear subspaces of $\reals^d$ and $\N_t=L^0(N_t,\F_t)$. 
\end{lemma}
\begin{proof}
We consider first the larger set 
\begin{align*}
\M_t=\big\{x_t\in L^0(\reals^d,\F_t)\mk\exists \,x_{t+1}&\pk x_{T}\:\text{s.t.}\:\Delta x_s\in L^0(C_s^0,\F_s)\\ &\forall\, s=t+1\pk T\:\text{and}\: x_T=0\big\} 
\end{align*}
and show that there is an $\F_t$-measurable map $M_t:\Omega\tos\reals^d$ whose values are linear subspaces of $\reals^d$ and $\M_t=L^0(M_t,\F_t)$. Then the $\F_t$-measurable map $N_t:\Omega\tos\reals^d$ with $N_t(\omega):=M_t(\omega)\cap C_t^0(\omega)$ has the desired properties. 

In order to obtain the maps $M_t$ we argue by induction on $T$. For $T=0$ we have $\M_T=M_T=\{0\}$. Now assume that the claim holds for any $T$-step model and consider the $T+1$-step model $(C_t)_{t=0}^{T}$. By the induction hypothesis there exist $\F_s$-measurable maps $M_s$ whose values are linear subspaces of $\reals^d$ such that $\M_s=L^0(M_s,\F_s)$ for each $s\in\{1\pk T\}$. Note that $x_0\in\M_0$ iff $x_0\in M_1+C_1^0$ almost surely. Indeed, $x_0\in\M_0$ iff there exists an $x_1\in\M_1$ such that  $x_1-x_0\in C_1^0$ almost surely. The mapping $L(\omega):=M_1(\omega)+C_1^0(\omega)$ is $\F_1$-measurable (see \cite[Proposition~14.11(c)]{rw98}) and its values are linear subspaces of $\reals^d$ and in particular closed. By theorem on page 135 of \cite{tr91}, there exists the largest closed $\F_0$-measurable set-valued map $M_0$ such that $M_0\subseteq L$ almost surely and $L^0(M_0,\F_0)=L^0(L,\F_0)=\M_0$. Clearly $M_0(\omega)$ is a linear subspace of $\reals^d$ for each $\omega$.
\end{proof}

For each $\omega$ we denote by $N_t^\perp(\omega)$ the orthogonal complement of $N_t(\omega)$ in $\reals^d$. Then $N_t^\perp:\Omega\tos\reals^d$ is an $\F_t$-measurable (see \cite[14.12(e)]{rw98}) map whose values are linear subspaces of $\reals^d$.

\begin{lemma}\label{lem:ort}
Let $C$ be a convex market model and let $c\in A(C)$. Then there exists a process $x^\perp\in\A$ such that $x_t^\perp\in N_t^\perp$, $\Delta x^\perp_t+c_t\in C_t$ for all $t=0\pk T$ and $x^\perp_T=0$.
\end{lemma}
\begin{proof}
We will prove the existence of the process $x^\perp$ by induction on $T$. For $T=0$ we have $x_T^\perp:=x_T=0$. Assume that the claim holds for any $T$-step model and consider the $T+1$-step model $(C_t)_{t=0}^T$. Let $c\in A(C)$ and let $x\in\A$ be such that $\Delta x_t+c_t\in C_t$ for all $t=0\pk T$ and $x_T=0$. We denote by $x_0^0$ the $\F_0$-measurable projection of $x_0$ on $N_0$ and by $(x_1^0\pk x_{T-1}^0,0)$ an extension of $x_0^0$ to a self-financing portfolio process in the model $C^0$ (the existence of such an extension is given by Lemma~\ref{lem:proj}). Then 
\[
y_t:=x_t-x^0_t,\quad t=0\pk T-1,\quad y_T:=0
\]
defines an adapted process $y$ with $y_0=x_0-x_0^0\in N_0^\perp$ almost surely. Moreover, $y$ hedges $c$. Indeed, since $\Delta x_t^0\in C^0_t$ we have
\[
\Delta y_t+c_t=\Delta x_t+c_t-\Delta x_t^0\in C_t-C_t^0=C_t\quad P\text{-a.s.}
\]
for all $t=0\pk T$. 

The process $(y_1,\pk y_{T-1},0)$ hedges the claim $(c_1-y_0,c_2\pk c_T)$ in the $T$-step model $(C_t)_{t=1}^T$. Thus by the induction hypothesis there is a process  $y^\perp=(y_1^\perp,\pk y_{T-1}^\perp,0)$ such that $y_t^\perp\in N_t^\perp$ almost surely and $y^\perp$ hedges $(c_1-y_0,c_2\pk c_T)$.  Then the process  $x^\perp:=(y_0, y_1^\perp,\pk y_{T-1}^\perp,0)$ has the required properties.
\end{proof}

The following lemma was used first in Kabanov and Stricker \cite{ks1}, it is also documented e.g.\ in \cite{sch4}, \cite{krs3}, \cite{fs4} and \cite{ds6}. We refer to these works for the proof.

\begin{lemma}\label{lem:sequence}
Let $(x^n)_{n=1}^\infty$ be a sequence of random vectors in $L^0(\Omega,\F, P, \reals^d)$ such that $\liminf_n|x^n|<\infty\;P$-almost surely. Then 
there exists an $\F$-measurable increasing $\naturals$-valued random sequence $(\tau^n)_{n=1}^\infty$ such that $(x^{\tau^n})_{n=1}^\infty$ converges almost surely to some $x\in L^0(\Omega,\F, P, \reals^d)$.
\end{lemma}

\noindent
{\bf Proof of Theorem \ref{thm:cl}.}
By Lemma~\ref{lem:ata}, the closedness of $A_T(C)$ is equivalent to the closedness of $A(C)$. By Lemma \ref{lem:ort}, we may assume that $x_t\in N^\perp_t$ for all $t=0\pk T-1$ in the definition of $A(C)$. 

In order to prove the closedness of $A(C)$ we will prove the following more precise statement:
Let $(c^n)_{n=1}^\infty$ be a sequence of claim processes in $A(C)$ such that $c^n\to c\in\A$ in measure and let $(x^n)_{n=1}^\infty$ be any sequence in $\A$ such that $\Delta x_t^n+c^n_t\in C_t,\;x^n_t\in N^\perp_t$ for all $t$ and $x^n_T=0$. Then there exists an $\F$-measurable random subsequence $(\tau^n)_{n=1}^\infty$ of $\naturals$ and an $x\in\A$ such that $x^{\tau^n}\to x$ almost surely, $\Delta x_t+c_t\in C_t,\;x_t\in N^\perp_t$ for all $t$ and $x_T=0$. In particular, $c\in A(C)$.

The proof will follow by induction on $T$. For $T=0$ the statement is obvious since the set $L^0(C_0,\F_0)$
is closed in $L^0(\Omega, \F_0, P, \reals^d)$ and any hedging strategy is identical $0$. Now assume that the statement holds for any $T$-step model and consider a $T+1$-step model $C=(C_t)_{t=0}^{T}$. Let $c^n\in A(C),\;n=0,1\pke$ be such that $c^n\to c\in\A$ in measure. By passing to a subsequence if necessary, we may assume that $c^n\to c$ almost surely. Let $x^n\in\A,\;n=0,1\pke$ be such that $\Delta x_t^n+c_t^n\in C_t,\;x^n_t\in N^\perp_t$  for all $t$ and $x^n_{T}=0$.

Case 1: the sequence $(x_0^n)_{n=1}^\infty$ is almost surely bounded. In this case, we can apply Lemma \ref{lem:sequence} to find an $\F_0$-measurable random sequence $(\sigma^n)_{n=1}^\infty$ such that $x_0^{\sigma^n}$  converges to some $x_0\in L^0(\Omega, \F_0, P, \reals^d)$ almost surely. Then $x_0\in L^0(N^\perp_0,\F_0)$ since each $N_0^\perp(\omega)$ is a closed subspace of $\reals^d$.  Moreover, the sequence of claim processes
\[
\tilde c^n:=(c^{\sigma^n}_1-x_0^{\sigma^n}, c^{\sigma^n}_2,\ldots, c^{\sigma^n}_{T}),\quad n\in\naturals
\]
belongs to $A\left((C_t)_{t=1}^{T}\right)$ with the hedging sequence $\tilde x^n:=(x_1^{\sigma^n}\pk x_{T-1}^{\sigma^n},0)$. Since $\tilde c^n\to(c_1-x_0,c_2\pk c_{T})$ almost surely, we apply the induction hypothesis to the $T$-step model $(C_t)_{t=1}^{T}$ and obtain an $\F$-measurable random subsequence $(\tau^n)_{n=0}^\infty$ of $(\sigma^n)_{n=0}^\infty$ and an $(\F_t)_{t=1}^{T}$-adapted process $\tilde x=(x_1\pk x_{T})$ such that $\tilde x^{\tau^n}\to\tilde x$ almost surely, $\tilde x$ hedges $\tilde c$ and $x_t\in N^\perp_t$ for $t=1,\ldots,T$.  This proves the claim, since $x_t^{\tau^n}\to x_t$ almost surely for all $t=0\pk T$ and the process $x:=(x_0\pk x_{T})$ has  the desired properties. Indeed, $x$ is adapted, $x_t\in N^\perp_t$ for all $t$, $x_{T}=0$ and we have 
\[
\Delta x_t+ c_t\in C_t\quad P\text{-a.s.}\quad \mbox{for all}\quad t=0\pk T
\]
since $C_t(\omega)$ is a closed subset of $\reals^d$ for each $\omega$.  

Case 2: the sequence $(x_0^n)_{n=1}^\infty$ is not almost surely bounded. We will show that this leads to a contradiction with the robust no scalable arbitrage property. Indeed, assume that the set
\[
A:=\{\omega\in\Omega\,|\,\liminf|x^n_{0}(\omega)|=\infty\}
\]
has positive probability and consider the sequences
\[
\hat x^n:=\alpha^nx^n\quad\mbox{and}\quad \hat c^n:=\alpha^nc^n,\qquad n\in\naturals,
\]
where $\alpha^{n}=\frac{\chi_{A}}{\max\{|x^{n}_{0}|,1\}}$. The processes $(\hat x^n)_{n=1}^\infty$ and $(\hat c^n)_{n=1}^\infty$ are adapted, $\hat x^n_t\in L^0(N^\perp_t,\F_t)$ for all $t$, $\hat x^n_{T}=0$, $\hat c_t^n\to0$ almost surely for each $t$ and the sequence $(\hat x_0^n)_{n=0}^\infty$ is almost surely bounded. Moreover, we have
\begin{equation}\label{ac}
\hat x^n_t -\hat x^n_{t-1}+ \hat c^n_t\in\alpha^n C_t\quad P\text{-a.s.},
\end{equation}
where $\alpha^n C_t\subset C_t$ since $C_t$ is convex, $0\in C_t$ and $|\alpha^n|\le1$.  Thus we have the same situation as in case 1, and using the same reasoning we obtain an $\F$-measurable random sequence $(\tau^n)_{n=0}^\infty$ and an adapted process $x$ such that $\hat x^{\tau^n}\to x$ almost surely, $x_t\in N^\perp_t$ for all $t$ and $x_{T}=0$. Moreover, since $\alpha^n\to 0$ a.s., \eqref{ac} and \eqref{recession} imply
\[
 \Delta x_t\in C_t^\infty\quad P\text{-a.s.}\quad \mbox{for all}\quad t=0\pk T.
\]
Thus $x\in\N(C^0)$, by Lemma~\ref{lem:zero}, and then $x_0\in N_0$ almost surely, by Lemma~\ref{lem:proj}. Since also $x_0\in N_0^\perp$, we have $x_0=0$ almost surely.  But we also have  $|\hat x^{n}_{0}|\to 1$ and hence  $|x_{0}|=1$ on the nontrivial set $A$, which is a contradiction. Thus, case 2 cannot occur under the robust no scalable arbitrage condition.\hfill $\square$

\bibliographystyle{plain}
\bibliography{sp}

\begin{thebibliography}{10}

\bibitem{ab99}
Charalambos~D. Aliprantis and Kim~C. Border.
\newblock {\em Infinite-dimensional analysis}.
\newblock Springer-Verlag, Berlin, second edition, 1999.
\newblock A hitchhiker's guide.

\bibitem{at7}
F.~Astic and N.~Touzi.
\newblock No arbitrage conditions and liquidity.
\newblock {\em J. Math. Econom.}, 43:692--708, 2007.

\bibitem{cjp4}
U.~{\c{C}}etin, R.~A. Jarrow, and P.~Protter.
\newblock Liquidity risk and arbitrage pricing theory.
\newblock {\em Finance Stoch.}, 8(3):311--341, 2004.

\bibitem{cr7}
U.~{\c{C}}etin and L.~C.~G. Rogers.
\newblock Modelling liquidity effects in discrete time.
\newblock {\em Mathematical Finance}, 17(1):15--29, 2007.

\bibitem{cst7}
U.~{\c{C}}etin, M.~H. Soner, and N.~Touzi.
\newblock Option hedging for small investors under liquidity costs.
\newblock {\em Preprint}, 2007.

\bibitem{ck96}
Jak{\v{s}}a Cvitani{\'c} and Ioannis Karatzas.
\newblock Hedging and portfolio optimization under transaction costs: a
  martingale approach.
\newblock {\em Math. Finance}, 6(2):133--165, 1996.

\bibitem{dkv2}
F.~Delbaen, Yu.~M. Kabanov, and E.~Valkeila.
\newblock Hedging under transaction costs in currency markets: a discrete-time
  model.
\newblock {\em Math. Finance}, 12(1):45--61, 2002.

\bibitem{ds6}
F.~Delbaen and W.~Schachermayer.
\newblock {\em The Mathematics of Arbitrage}.
\newblock Springer Finance. Springer-Verlag, Berlin Heidelberg, 2006.

\bibitem{fs4}
H.~F{\"o}llmer and A.~Schied.
\newblock {\em Stochastic finance}, volume~27 of {\em de Gruyter Studies in
  Mathematics}.
\newblock Walter de Gruyter \& Co., Berlin, extended edition, 2004.
\newblock An introduction in discrete time.

\bibitem{js98}
J.~Jacod and A.~N. Shiryaev.
\newblock Local martingales and the fundamental asset pricing theorems in the
  discrete-time case.
\newblock {\em Finance Stoch.}, 2(3):259--273, 1998.

\bibitem{jk95a}
E.~Jouini and H.~Kallal.
\newblock Martingales and arbitrage in securities markets with transaction
  costs.
\newblock {\em J. Econom. Theory}, 66(1):178--197, 1995.

\bibitem{jns5}
E.~Jouini, C.~Napp, and W.~Schachermayer.
\newblock Arbitrage and state price deflators in a general intertemporal
  framework.
\newblock {\em J. Math. Econom.}, 41(6):722--734, 2005.

\bibitem{kab99}
Yu.~M. Kabanov.
\newblock Hedging and liquidation under transaction costs in currency markets.
\newblock {\em Finance and Stochastics}, 3(2):237--248, 1999.

\bibitem{kab1}
Yu.~M. Kabanov.
\newblock Arbitrage theory.
\newblock In {\em Option pricing, interest rates and risk management}, Handb.
  Math. Finance, pages 3--42. Cambridge Univ. Press, Cambridge, 2001.

\bibitem{krs2}
Yu.~M. Kabanov, M.~R{\'a}sonyi, and Ch. Stricker.
\newblock No-arbitrage criteria for financial markets with efficient friction.
\newblock {\em Finance Stoch.}, 6(3):371--382, 2002.

\bibitem{ks1b}
Yu.~M. Kabanov and Ch. Stricker.
\newblock The {H}arrison-{P}liska arbitrage pricing theorem under transaction
  costs.
\newblock {\em J. Math. Econom.}, 35(2):185--196, 2001.
\newblock Arbitrage and control problems in finance.

\bibitem{ks1}
Yu.~M. Kabanov and Ch. Stricker.
\newblock A teachers' note on no-arbitrage criteria.
\newblock In {\em S\'eminaire de Probabilit\'es, XXXV}, volume 1755 of {\em
  Lecture Notes in Math.}, pages 149--152. Springer, Berlin, 2001.

\bibitem{krs3}
Yuri Kabanov, Mikl{\'o}s R{\'a}sonyi, and Christophe Stricker.
\newblock On the closedness of sums of convex cones in {$L\sp 0$} and the
  robust no-arbitrage property.
\newblock {\em Finance Stoch.}, 7(3):403--411, 2003.

\bibitem{km6}
K.~Kaval and I.~Molchanov.
\newblock Link-save trading.
\newblock {\em J. Math. Economics}, 42:710--728, 2006.

\bibitem{pen6}
T.~Pennanen.
\newblock Nonlinear price processes.
\newblock {\em Preprint}, 2006.

\bibitem{pen8}
T.~Pennanen.
\newblock Arbitrage and deflators in illiquid markets.
\newblock {\em Submitted}, 2008.

\bibitem{pen8b}
T.~Pennanen.
\newblock Superhedging in illiquid markets.
\newblock {\em Submitted}, 2008.

\bibitem{ipen1}
I.~Penner.
\newblock Arbitragefreiheit in \protect{F}inanzm{\"a}rkten mit
  \protect{T}ransaktionskosten.
\newblock Master's thesis, Humboldt-Universit\"at zu Berlin, 2001.

\bibitem{ras7}
M.~R{\'a}sonyi.
\newblock New methods in the arbitrage theory of financial markets with
  transaction costs.
\newblock {\em Forthcoming in S{\'e}minaire de Probabilit{\'e}s}, 2007.

\bibitem{roc70a}
R.~T. Rockafellar.
\newblock {\em Convex analysis}.
\newblock Princeton Mathematical Series, No. 28. Princeton University Press,
  Princeton, N.J., 1970.

\bibitem{rw98}
R.~T. Rockafellar and R.~J.-B. Wets.
\newblock {\em Variational analysis}, volume 317 of {\em Grundlehren der
  Mathematischen Wissenschaften [Fundamental Principles of Mathematical
  Sciences]}.
\newblock Springer-Verlag, Berlin, 1998.

\bibitem{rs7}
L.~C.~G. Rogers and S.~Singh.
\newblock The cost of illiquidity and its effects on hedging.
\newblock {\em Preprint}, 2007.

\bibitem{rok7}
D.~B. Rokhlin.
\newblock Martingale selection problem and asset pricing in finite discrete
  time.
\newblock {\em Electron. Comm. Probab.}, 12:1--8, 2007.

\bibitem{sch4}
W.~Schachermayer.
\newblock The fundamental theorem of asset pricing under proportional
  transaction costs in finite discrete time.
\newblock {\em Math. Finance}, 14(1):19--48, 2004.

\bibitem{tr91}
A.~Truffert.
\newblock Conditional expectation of integrands and random sets.
\newblock {\em Ann. Oper. Res.}, 30(1-4):117--156, 1991.
\newblock Stochastic programming, Part I (Ann Arbor, MI, 1989).

\end{thebibliography}
\end{document}